\documentclass[10pt]{eptcs}
\usepackage[english]{babel}
\usepackage[T1]{fontenc}
\usepackage[utf8]{inputenc}
\usepackage{cite}
\usepackage{amsmath}
\usepackage{mathtools}
\usepackage{amssymb}
\usepackage{amsthm}
\usepackage{graphicx}
\usepackage{float}
\usepackage{enumerate}
\usepackage{txfonts}

\title{ Topology and Non-Deterministic Polynomial Time Computation : Avoidance of The Misbehaviour of Hub-Free Diagrams and Consequences}

\author{Anthony Gasperin
\institute{Theoretical Computer Science Department,University of Geneva, Switzerland }
\email{\quad anthony.gasperin@unige.ch}
}

\def\titlerunning{Avoidance of The misbehaviour of Hub-Free Diagrams and Consequences}
\def\authorrunning{A. Gasperin}
\begin{document}
\maketitle

\newtheorem{Def}{Definition}
\newtheorem{Lem}{Lemma}
\newtheorem{Th}{Theorem}
\newtheorem{Conj}{Conjecture}
\newtheorem{Prop}{Proposition}
\newtheorem{Cor}{Corollary}
\newtheorem{Claim}{Claim}
\newtheorem{Rem}{Remark}

\begin{abstract}
To study groups with small Dehn's function, Olshanskii and Sapir developed a new invariant of bipartite chords diagrams and applied it to 
hub-free realization of $\mathcal{S}$-machines.
In this paper we consider this new invariant together with groups constructed from $\mathcal{S}$-machines containing the hub relation.
The idea is to study the links between the topology of the asymptotic cones and polynomial time computations.
Indeed it is known that the topology of such metric space depends on diagrams without hubs that do not correspond to the computations of the
considered $\mathcal{S}$-machine. This work gives sufficient conditions that avoid this misbehaviour,
but as we shall see the method has a significant drawback.

\end{abstract}
\section{Introduction}
Let $(X,d_X)$ a metric space $s=(s_n)$ a sequence of points in $X$, $d=(d_n)$ an increasing sequences of numbers with 
$\text{lim }d_n=\infty$ and let $\omega: P(\mathbb{N}) \to \{0,1\}$ be a non-principal ultrafilter. An \textit{asymptotic cone} of
$Con_\omega(X,s,d)$ of $(X,d_X)$ is the subset of the cartesian product $X^\mathbb{N}$ consisting of sequences $(x_i)_{i \in \mathbb{N}}$
with $\text{lim}_\omega \frac{d_X( s_i,x_i)}{d_i} < \infty$ where two sequences $(x_i)$ and $(y_i)$ are equivalent if and only if
$\text{lim}_\omega \frac{d_x(x_i,y_i)}{d_i}=0$. The distance between two elements $(x_i),(y_i)$ in the asymptotic cone 
$Con_\omega(X,s,d)$ is defined as $\text{lim}_\omega \frac{d_X(x_i,y_i)}{d_i}$. Here $\text{lim}_\omega$ is defined as follows. If $a_n$
is a bounded sequence of real numbers then $\text{lim}_\omega(a_n)$ is the unique number $a$ such that for every $\epsilon >0$,
$\omega(\ \{ n\ |\ |a_n-a| < \epsilon \}\ )=1$.
The \emph{asymptotic cones} of a finitely generated group $G$ are asymptotic cones of the Cayley graph of $G$ and it is well known that 
they do not depend on the choice of the sequence $s$. It is then assumed that $s=(1)$ where $1$ is the identity. 
Given an ultrafilter $\omega$ and an increasing sequence of numbers $d$ the asymptotic cone of a finitely generated group $G$ is then denoted $Con_\omega(X,d)$.

A function $f:\mathbb{N} \to \mathbb{N}$ is an isoperimetric function of a finite presentation $\langle X,R \rangle$ of a group $G$ if every word $w$ in $X$,
that is equal to $1$ in $G$, is freely equal to a product of conjugates $\prod^{m}_{i=1} x_i^{-1} r_i x_i$ where $r_i$ or $r^{-1}_i$ is in $R$, $x_i$ is
in $(X \cup X^{-1})^*$ and $m \leq f(|w|)$.
The \emph{Dehn's function} of a finite presentation $\langle X,R \rangle$ is defined as the smallest \emph{isoperimetric} function of the presentation.\\
In \cite{Bir02,Gro93,DBLP:journals/ijac/OiSHANSKIIS01} the connections between Dehn's functions, asymptotic geometry of groups and computational complexity of the word problem are discussed.
In \cite{Gro93} Gromov showed that if all \emph{asymptotic cones} of a group $G$ are simply connected then $G$ is finitely presented, has polynomial isoperimetric
function and linear isodiametric function. Papasoglu \cite{Pap96} proved that if a finitely presented group has quadratic isoperimetric function then all its
asymptotic cones are simply connected. 
Sapir, Birget and Rips in \cite{Sap02} introduced the concept of $\mathcal{S}$-machines to show that the word problem of a finitely generated group is decidable
in polynomial time if and only if this group can be embedded into a group with polynomial isoperimetric function.
Olshanskii and Sapir in \cite{Sap05} constructed a group with polynomial isoperimetric function, linear 
isodiametric function and non-simply connected asymptotic cones, the group is roughly an $\mathcal{S}$-machine introduced in \cite{Sap02}. In \cite{DBLP:journals/ijac/OlshanskiiS07} 
they also constructed a group with two non-homeomorphic asymptotic cones  using the concept of $\mathcal{S}$-machine.
Considering hub-free realization of $\mathcal{S}$-machines \cite{Ols06} developed a new invariant of bipartite chords diagrams. In particular this invariant
is used to prove that there exists a finitely presented multiple HNN extension of free groups with Dehn's function $n^2 log\ n$ and undecidable conjugacy problem.
Recently \cite{DBLP:conf/cie/Gasperin13}, using results and suggestions of \cite{Sap05}, studied the links between polynomial time computations and topology.
The work focuses on a direct construction of an $\mathcal{S}$-machine from an $NP$-complete problem. \cite{DBLP:conf/cie/Gasperin13} emphasizes that if one wants to study 
the topology of the asymptotic cones of some non-deterministic computation then one has to be careful when constructing the $\mathcal{S}$-machine.
Indeed the topology highly depends on the construction of the considered $\mathcal{S}$-machine and diagrams that do not correspond to the computation of the machine. 
In particular, the consideration of hub-free diagrams and some explicit rules, allows to prove that all asymptotic cones of the groups are not simply connected.

In this paper we consider groups with the hub relation together with the new invariant of bipartite chords diagram. 
We show that in a such setting it is possible to construct a group such that the hub-free diagrams satisfy a Lemma stated in \cite{DBLP:journals/ijac/OlshanskiiS07}. 
Then it comes that such diagrams cannot be used anymore to contradict the simply connectivity of the asymptotic cones. 
As a major drawback the area of diagrams containing the hub relation and corresponding to computation is increased exponentially.

\section{Preliminaries}  
This section introduces briefly the machinery introduced by Sapir, Birget and Rips in \cite{Sap02}.
We need to explain, at least superficially, what is an $S$-machine, how it works and especially how it leads to the construction of group.

\subsection{$S$-machines}

This section is closely modeled on \cite{Sap02}, we recall the notion of $S$-machine defined in the work of Sapir, Birget and Rips in \cite{Sap02}.

\cite{Sap02} defined $S$-machines as rewriting systems. An $S$-machine then comes with a \textit{hardware}, a \textit{language of admissible words}, and
a set of \textit{rewriting rules}.
A \textit{hardware } of an $S$-machine is a pair $(Y,Q)$ where $Y$ is an $n$-vector of (not necessarily disjoint) sets $Y_i$, $Q$ is an $(n+1)$-vector of 
disjoints sets $Q_i$ with $(\bigcup Y_i) \cap (\bigcup Q_i) = \emptyset$.
The elements of $\bigcup Y_i$ are called \textit{tape letters}; the elements of $\bigcup Q_i$ are called \textit{state letters}.
With every hardware $\mathcal{S}=(Y,Q)$ one can associate the \textit{language of admissible words} $L(\mathcal{S})=Q_1 F(Y_1) Q_2 \cdots F(Y_n)Q_{n+1}$ where
$F(Y_i)$ is the language of all reduced group words in the alphabet $Y_j \cup Y^{-1}_j$.
This language completely determines the hardware. One can then describe the language of admissible words instead of describing the hardware $\mathcal{S}$.
If $1 \leq i < j \leq n$ and $W=q_1 u_1 q_2 \cdots u_n q_{n+1}$ is an admissible word, $q_i \in Q_i, u_i \in (Y_i \cup Y^{-1}_i)^*$ then the 
subword $q_i u_i \cdots q_j$ of $W$ is called the $(Q_i,Q_j)$-subword of $W$ ($i < j)$.
The rewriting rules ( $S$-\textit{rules}) have the following form:
\begin{center}
	$[U_1 \to V_1, \dots, U_m \to V_m ]$
\end{center}
where the following conditions hold:
Each $U_i$ is a subword of an admissible word starting with a $Q_l$-letter and ending with $Q_r$-letter.
If $i < j$ then $r(i) < l(j)$, where $r(i)$ is the end of $U_i$ and $l(j)$ the start of $U_j$.
Each $V_i$ is a subword of and admissible word whose $Q$-letters belong to $Q_{l(i)} \cup \cdots \cup Q_{r(i)}$.
The machine applies an $S$-rule to a word $W$ replacing simultaneously subwords $U_i$ by subword $V_i,i=1,\dots,m$.

As mentioned in \cite{Sap02} there exists a natural way to convert a Turing machine $M$ into an $S$-machine $\mathcal{S}$; 
one can concatenate all tapes of the given machine $M$ together and replace every command $aq\omega \to q' \omega$ by $a^{-1}q'\omega$. 
Unfortunately the $S$-machine constructed following this
natural way will not inherit most of the properties of the original machine $M$.
According to \cite{Sap02} the main problem is that it is nontrivial to construct an $S$-machine which recognizes only positive powers 
of a letter.
Thus in order to construct an $S$-machine $\mathcal{S}(M)$ that will inherit the desired properties of a Turing machine $M$, Sapir, Birget and Rips 
in \cite{Sap02} constructed eleven $S$-machines and then used them to construct the final $S$-machine $\mathcal{S}(M)$ simulating $M$. 
The construction is quite involved and nontrivial, one can see \cite{Sap02} for details.

Taking any Turing machine $M$ and modifying it in a specific way, \cite{Sap02}
constructs an $S$-machine $\mathcal{S}(M)$ simulating $M$. The $S$-machine constructed in \cite{Sap02} is quite long to define, next we explain 
briefly the main part of the construction, for proofs and deeper understanding of the whole machinery please refer to \cite{Sap02}. 
The main idea of the construction is to simulate the initial machine $M$ using eleven $S$-machines $S_1,S_2, \dots,S_9,S_{\alpha},S_{\omega}$.
We will explain how the machines  $S_4,S_9,S_{\alpha},S_{\omega}$ are used in the construction of $\mathcal{S}(M)$. The others $S$-machines are used to construct 
$S_4$ and $S_9$ and are rather of technical importance.
First we need to describe what is an admissible word of the $S$-machine $\mathcal{S}(M)$.
For every $q \in Q$ the word $q\omega$ is denoted by $F_q$, in every command of $M$ the word $q\omega$ is replaced by $F_q$.
Left marker on tape $i$ is denoted by $E_i$. This gives a Turing Machine $M'$ such that the configurations of each tape have the form 
$E_i u F_q$ where $u$ is a word in the alphabet of tape $i$ and every command or its inverse has one of the forms:
\begin{equation}\label{positive_rules}
	\{ F_{q_1} \to F_{q'_1}, \dots, aF_{q_i} \to F_{q'_i}, \dots,F_{q_k}\to F_{q'_k} \}
\end{equation}
where $a \in Y$ or 

\begin{equation}
	\{ F_{q_1} \to F_{q'_1}, \dots, E_i F_{q_i} \to E_i F_{q'_i} ,\dots , F_{q_k} \to F_{q'_k} \}.
\end{equation}

An admissible word of the considered $\mathcal{S}(M)$ machine is a product of three parts. 
The first part has the form 
\begin{center}
	$E(0) \alpha^{n_1} x(0) \alpha^{n_2} F(0)$.
\end{center}
The second part is a product of $k$ words of the form
\begin{center}
	$E(i)v_ix(i)w_iF(i)E'(i) p(i) \Delta_{i,1} q(i) \Delta_{i,2}r(i) \Delta_{i,3} s(i) \Delta_{i,4} t(i) \Delta_{i,5}$\\ 
	$u(i) \Delta_{i,6} \overline{p}(i) \Delta_{i,7} \overline{q}(i) \Delta_{i,8} \overline{r}(i) \Delta_{i,9}$\\
	$\overline{s}(i) \Delta_{i,10} \overline{t}(i)\Delta_{i,11} \overline{u}(i) \Delta_{i,12} F'(i), i =1,\dots,k$
\end{center}
The third part has the form 
\begin{center}
	$E'(k+1) \omega^{n'_1}x'(k+1)\omega^{n'_2}F'(k+1)$.
\end{center}
Here $v_i,w_i$ are group words in the alphabet $Y_i$ of tape $i$, and $\Delta_{i,j}$ is a power of $\delta$. The letters 
\begin{center}
$E(i),x(i),F(i),E'(i),p(i),q(i),r(i),s(i),t(i),u(i),\overline{p}(i),\overline{q}(i),\overline{r}(i)$,
$\overline{s}(i),\overline{t}(i),\overline{u}(i),F'(i)$ 
\end{center}
belong to disjoint sets of state letters.
The letters $x(i),p(i),q(i),r(i),s(i),t(i),u(i),\overline{p}(i),\overline{q}(i),\overline{r}(i)$,
$\overline{s}(i), \overline{t}(i), \overline{u}(i)$
are called standard and are included into the corresponding sets 
$\bold{X}(i),\bold{P}(i),\bold{R}(i),\bold{S}(i),\bold{T}(i),\bold{U}(i),\overline{\bold{P}}(i),\overline{\bold{Q}}(i)$
,$\overline{\bold{R}}(i),\overline{\bold{S}}(i),\overline{\bold{T}}(i),\overline{\bold{U}}(i),(i=1,\dots,k)$.
Let $\tau$ be a command in $\Theta$ of the form (\ref{positive_rules}) ( command of the form (\ref{positive_rules}) are called \textit{positive } and their inverse \textit{negative}).
For every $\gamma \in \{4,9,\alpha,\omega\}$ and for each component $V(i)$ of the vector of sets of state letters, the letters 
$V(i,\tau,\gamma)$ are included into $V(i)$ where $V \in \{ P, Q,R,S,T,U,\overline{P},\overline{Q},\overline{R},\overline{S},\overline{T},\overline{U}\}$.
For each $S$-machine $S_{\gamma}, \gamma \in\{4,9,\alpha,\omega\}$ a copy of $S_{\gamma}$ is considered where every state letter $z$ is replaced
by $z(j,\tau,\gamma)$ where $j=i$ if $\gamma=4,9, j=0$ if $\gamma=\alpha$ and $j=k+1$ if $\gamma=\omega$. 
These state letters are included into the corresponding sets. The state letters we just described are all the state letters of $\mathcal{S}(M)$.
The rules of $\mathcal{S}(M)$ are the rules of $\mathcal{S}_4(\tau),\mathcal{S}_9(\tau),\mathcal{S}_{\gamma}(\tau),\mathcal{S}_{\omega}(\tau)$
for all $\tau \in \Theta$ of the form (\ref{positive_rules}) plus the connecting rules.
Basically the connecting rules allow to go from a machine to another one, there are five such rules:
$R_4(\tau),R_{4,\alpha}(\tau),R_{\alpha,\omega}(\tau),R_{\omega,9}(\tau),R_9(\tau)$.
They can be described informally as follows.
$R_4(\tau)$ turns on the machine $\mathcal{S}_4(\tau)$. $R_{4,\alpha}(\tau)$ turns on the machine $\mathcal{S}_{\alpha}(\tau)$ when 
$\mathcal{S}_4(\tau)$ finishes its work, $R_{\alpha,\omega}(\tau),R_{\omega,9}(\tau)$ do the same with the corresponding $S$-machines.
$R_9(\tau)$ turns off $\mathcal{S}_9(\tau)$ and gets the machine ready to simulate the next transition from $\Theta$.
This machinery contains all the necessary steps to simulate a rule of the machine $M$.

Formally speaking, to every configuration $c=(E_1 v_1F_{q_1},\dots,E_k v_k F_{q_k})$ of the machine $M$ is associated the following
admissible word $\sigma(c)$ of $\mathcal{S}(M)$:\\
$E(0) \alpha^n x(0) F(0)$\\
$E(1)v_1 x(1) F_{q_1}(1)E'(1)p(1)\delta^{||v_1||}q(1)r(1)s(1)t(1)u(1)$\\
$\overline{p}(1) \overline{q}(1) \overline{r}(1) \overline{s}(1)\overline{t}(1)\overline{u}(1)F'_{q_1}(1)\dots$\\
$E(k)v_k x(k) F_{q_k}(k)E'(k)p(k)\delta^{||v_k||}q(k)r(k)s(k)t(k)u(k) $\\
$\overline{p}(k) \overline{q}(k) \overline{r}(k) \overline{s}(k)\overline{t}(k)\overline{u}(k)F'_{q_k}(k)$\\
$E'(k+1)x'(k+1)\omega^n F'(k+1)$, where $||v||$ is the algebraic sum of the degree of the letters in $v$.

Now we present how \cite{Sap02} converts an $S$-machine $\mathcal{S}(M)$ into a group presentation, once again
this part is strongly modeled on \cite{Sap02}.
Let $\mathcal{S}(M)$ be the $S$-machine as constructed. Let $Y$ be the vector of sets of tape letters, and let $Q$ be the vector
of state letters of $S$. One can remark that $Q$ has $17k+6$ components which \cite{Sap02} denotes by $Q_1,\dots,Q_{17k+6}$. 
\cite{Sap02} notices that $Q_1=\textbf{E}(0),Q_2=\textbf{X}(0),Q_3=\textbf{F}(0), Q_{17k+4}=\textbf{E}'(k+1), Q_{17k+5}=\textbf{X}(k+1),Q_{17k+6}=\textbf{F}'(k+1)$.
Let $\mathbf\Theta_+$ the set of positive rules of $\mathcal{S}(M)$ and $N$ a positive integer. 
To construct their group $G_N(\mathcal{S})$ Sapir, Birget and Rips take the following generating sets :
\begin{equation}
	A=\bigcup\limits^{17k+6}_{i=1} Q_i \cup \{ \alpha,\omega,\delta \} \cup \bigcup\limits_{i=1}^{k} Y_i \cup \{ \kappa_j | j=1,\dots,2N \} \cup \mathbf\Theta_+.
\end{equation}
and the following set $P_N(\mathcal{S})$ of relations:
\begin{enumerate}
	\item \textit{Transitions relations}. These relations correspond to elements of $\mathbf\Theta_+$.
	Let $\tau \in \mathbf\Theta_+, \tau=[U_1 \to V_1,\dots,U_p \to V_p]$. Then relations $\tau^{-1}U_1\tau=V_1,\dots,\tau^{-1}U_p\tau=V_p$ are included 
	into $P_N(\mathcal{S})$. If for some $j$ from $1$ to $17k+6$ the letters from $Q_j$ do not appear in any of the $U_i$ then the relations
	$\tau^{-1} q_j \tau =q_j$ for every $q_j \in Q_j$ are also included.

	\item \textit{Auxiliary relations}. These are all possible relations of the form $\tau x= x \tau$ where 
	$x \in \{\alpha,\omega,\delta \} \cup \bigcup_{i=1}^k Y_i,\tau \in \mathbf\Theta_+$.

	\item \textit{The hub relation}. For every word $u$ let $K(u)$ denote the following word:
	\begin{center}
		$K(u) \equiv (u^{-1} \kappa_1 u \kappa_2 u^{-1} \kappa_3 u \kappa_4 \dots u^{-1} \kappa_{2N-1} u \kappa_{2N})\times $\\ 
		$(\kappa_{2N} u^{-1} \kappa_{2N-1}u \dots \kappa_2u^{-1}\kappa_1 u)^{-1}$.
	\end{center}
	Then the relation hub is $K(W_0)=1$, where $W_0$ is the accepting configuration of the $\mathcal{S}$-machine.
\end{enumerate}
The objective  of Sapir, Birget and Rips \cite{Sap02} in constructing such groups is to prove the following theorem:
\begin{Th} {\cite{Sap02}} {\label{supertheorem}}
Let $L \subseteq X^{+}$ be a language accepted by a Turing machine $M$ with a time function $T(n)$ for which $T(n)^4$ is superadditive.
Then there exists a finitely presented group $G(M)=\langle A \rangle$ with Dehn's function equivalent to $T(n)^4$, the smallest, isodiametric function
equivalent to $T^3(n)$, and there exists an injective map $K: X^+ \to (A \cup A^{-1})^+$ such that 
\begin{enumerate}
	\item $u \in L$ if and only if $K(u)=1$ in $G$;
	\item $K(u)$ has length $O(|u|)^2$ and is computable in time $O(|u|)$.
\end{enumerate}
\end{Th}

\section{ Adding $\mathcal{S}$-machine and composition}
In this section we briefly explain how the composed $\mathcal{S}$-machine is constructed in \cite{Ols06}.
We begin with the definition of the adding $\mathcal{S}$-machine. This section is closely modeled following the construction given in \cite{Ols06}.
\subsection{ The adding $\mathcal{S}$-machine}
Let $A$ be a finite set of letters, $A_0,A_1$ be a copy of $A_0$. For every $a_0 \in A_0$ let $a_1$ denote its copy in $A_1$.
\cite{Ols06} defines the adding $\mathcal{S}$-machine $Z(A)$ as follows.
Its state of letters is $P_1 \cup P_2 \cup P_3$ where $P_1=\{L \}, P_2=\{p(1),p(2),p(3) \}, P_3=\{R \}$. The set of tape
letters is $Y_1 \cup Y_2$ where $Y_1=A_0 \cup A_1$ and $Y_2=A_0$.
The positive rules of $Z(A)$ are the following:
\begin{itemize}
	\item $r_1(a)=[L \to L, p(1) \to a^{-1}_1 p(1) a_0, R \to R ]$.
	\item $r_{12}(a)=[L \to L, p(1) \to a^{-1}_0 a_1 p(2), R \to R ]$.
	\item $r_2(a) = [L \to L,p(2) \to a_0 p(2) a^{-1}_0, R \to R]$.
	\item $r_{21}= [L \to L,p(2) \to p(1), R \to R], Y_1(r_{21})=Y_1, Y_2(r_{21})=\emptyset$.
	\item $r_{13}=[L \to L, p(1) \to p(3), R \to R], Y_1(r_{13})=\emptyset, Y_2(r_{13})=A_0$.
	\item $r_3(a)= [L \to L,p(3) \to a_0p(3)a^{-1}_0, R \to R], Y_1(r_3(a))=Y_2(r_3(a))=A_0$.
\end{itemize}

\subsection{Composition of $\mathcal{S}$-machines}
Let us recall how \cite{Ols06} defines the composition of an arbitrary $\mathcal{S}$-machine and the adding $\mathcal{S}$-machine.
Essentially \cite{Ols06} defines the composition $\mathcal{S} \circ Z$ of $\mathcal{S}$ and $Z(A)$ inserting a $p$-letter between any two consecutive
$k$-letters in admissible words of $\mathcal{S}$, and treating any subword $k_i \dots p \dots k_{i+1}$ as an admissible word for $Z(A)$.
The set of state letters is defined as $K_1 \cup P_1 \cup K_2 \cup P_2 \cup \dots \cup P_{N-1} \cup K_N$ where 
$P_i=\{ p_i,p_i(\theta,1),p_i(\theta,3) |\ \theta \in \Theta \},i=1, \dots,N-1$. 
The set of tape letters is $\overline{Y}=(Y_{1,0} \cup Y_{1,1}) \cup Y_{1,0} \cup (Y_{2,0} \cup Y_{2,1}) \cup Y_{2,0} \cup \dots \cup (Y_{N-1,0} \cup Y_{N-1,1}) \cup Y_{N-1,0}$
where $Y_{i,0}, Y_{i,1}$ are copies of $Y_i$. The components of this union are denoted by $\overline{Y}_1,\dots,\overline{Y}_{2N-2}$.
The set of positive rule $\overline{\Theta}$ of $\mathcal{S} \circ Z$ is defined as a union of the set of modified positive rule of $\mathcal{S}$ and 
$2(N-1)|\Theta|$ copies $Z_i(\theta)^+$ $(\theta \in \Theta, i=1,\dots,N)$ of positive rules of the machine $Z(Y_i)$.
The idea is to slow down the working of $\mathcal{S}$, that is in order to simulate a computation of $\mathcal{S}$ consisting of sequence of rules 
$\theta_1,\dots,\theta_s$, first all rules corresponding to $\theta_1$ are applied , then all rules corresponding to $\theta_2$, etc.

Formally every positive rules $\theta \in \Theta^+$ of the form 
\begin{center}
	$[k_1u_1 \to k'_1 u'_1, v_1k_2u_2 \to v'_1 k'_2u'_2, \dots, v_{N-1}k_N \to v'_{N-1}k'_N]$,
\end{center}
where $k_i,k'_i \in K_i$, $u_i$ and $v_i$ are words in $Y$, is replaced by 
\begin{center}
	$\overline{\theta}=[k_1u_1 \to k'_1 u'_1, v_1p_1 \to v'_1p_1(\theta,1), k_2 u_2 \to k'_2 u'_2,\dots,v_{N-1}p_{N-1} \to v'_{N-1} P_{N-1}(\theta,1),k_N\to k'_N]$
\end{center}

with $\overline{Y}_{2i-1}(\overline{\theta})=Y_{i,0}(\theta)$ and $Y_{2i}=\emptyset$ for every $i$.
Informally each modified rule from $\Theta$ turns on $N-1$ copies of the machine $Z(A)$. 

\cite{Ols06} defines each machine $Z_i(\theta)$ as a copy of the machine $Z(Y_i)$ where every rule 
$\tau=[U_1 \to V_1, U_2 \to V_2, U_3 \to V_3]$ is replaced by the rule of the form 
\begin{center}
	$\overline{\tau}_i(\theta)=[ \overline{U}_1 \to \overline{V}_1,\overline{U}_2 \to \overline{V}_2, \overline{U}_3 \to \overline{V}_3,
				   k'_j \to k'_j, p_j(\theta,3) \to p_j(\theta),3), j=1,\dots,i-1, p_s(\theta,1)\to p_s(\theta,1),k'_{s+1} \to k'_{s+1} ,
				   s=i+1,\dots,N-1]$
\end{center}
where $\overline{U}_1, \overline{U}_2, \overline{U}_3, \overline{V}_1,\overline{V}_2,\overline{V}_3$ are obtained from $U_1,U_2,U_3,V_1,V_2,V_3$
by replacing $p(j)$ with $p_i(\theta,j)$, $L$ with $k'_i$ and $R$ with $k'_{i+1}$, and for $s\neq i$, $\overline{Y}_{2s-1}(\overline{\tau}_i(\theta))=Y_{i,0}$.
To work the machine needs a rule that returns all $p$-letters to their original form, \cite{Ols06} called it the transition rule.
\begin{center}
	$[k'_i \to k'_i, p_j(\theta,3) \to p_j, i=1,\dots,N-1]$.
\end{center}

\section{Construction of the group $\mathcal{S}\circ \mathcal{Z}$}

Remember that our construction is not exactly the same as the one presented in \cite{Ols06}. 
Indeed \cite{Ols06} considers hub-free realization of $\mathcal{S}$-machines. 
Our objective is to study the relation 
between the topology of the metric space constructed from an $\mathcal{S}$-machine working in polynomial time. 
Thus we consider the whole construction of \cite{Sap02} , that is we do not delete the hub relation of the group's presentation.
We shall prove two important things. The first is that the method developed by \cite{Ols06} together with sufficient conditions allow to avoid the misbehaviour 
of the diagram without hubs of the considered group. The second thing we prove is that the method is not ideal for our needs since it changes the polynomial nature 
of the diagram containing the hub relation. The reader should remark that we do not define the whole
machinery of \cite{Ols06}, in particular the concept of diagram dispersion $\mathcal{E}(\Delta)$. 
Indeed the proof in \cite{Ols06} is based on some crucial lemma, and then 
one needs to prove a variant of the Lemma to obtain directly a variant of the result, roughly this is what we shall do in the following.
Let us recall the concept defined in \cite{Ols06} of $\mathcal{B}$-covered base, $\mathcal{B}$-tight base and $\mathcal{B}$-narrow base.
Every word in the alphabet $\{Q_1,\dots,Q_{2N-1}\}$ is called a \emph{base word}. Let $\mathcal{B}$ be a finite set of base words.

\begin{Def} A base word is $\mathcal{B}$-covered if 
	\begin{itemize} 
		\item it is covered by bases from $\mathcal{B}$ (i.e every letter belongs to a subword from $\mathcal{B}$);
		\item it starts and ends with the same $q$-letter $x$.
	\end{itemize}
\end{Def}

A base word is called $\mathcal{B}$-\emph{tight} if it has the form $uxvx$ where $xvx$ is a $\mathcal{B}$-covered word, $w$ does not contain any other 
$\mathcal{B}$-covered subwords. A base word is called $\mathcal{B}$-\emph{narrow} if it does not contain $\mathcal{B}$-covered subwords.
We now begin our construction, we first need a machine that looks like the adding machine previously defined. Then we will make assumptions 
about the working of $\mathcal{S} \circ \mathcal{Z}$. Even if the construction of such a machine is not trivial, 
one can note that the assumption is not a big one since \cite{DBLP:journals/ijac/Olshanskii07} already proved that such a machine exists.

Let $\mathcal{Z}(A)$ be an $\mathcal{S}$-machine constructed in the same fashion as the adding machine and satisfying the following Lemma:

\begin{Lem}\label{explem}
	Let $base(W)=LpR$. Then for every computation $W=W_0 \to W_1 \to \dots \to W_t=f.W$ of the $\mathcal{S}$-machine $\mathcal{Z}(A)$:
	\begin{enumerate}
		\item $|W_i| \leq max(|W|, |f.W|), i=0,\dots,t$,
		\item If $W=LupR$ where $p=p(1)$ (resp. $p=p(3))$,$f.W$ contains $p(3)R$ (resp. $p(1)R)$ and all $a$-letters in $W,f.W$ are from
		$A^{\pm}_0$ ,then the length $g(|u|)$ of $f$ is between $2^{|u|}$ and $6 \cdot 2^{|u|}$, if $u$ is a positive word, and all words in the computation have the same length.
		Vice versa, for every positive word $u$, such a computation exist.
	\end{enumerate}
\end{Lem}

Moreover we assume that the composition $\mathcal{S} \circ \mathcal{Z}$, as defined previously, satisfies the two following Lemmas. 
\begin{Lem} \label{powerful-lemma} There exists a finite set of base words $\mathcal{B}$ such that 
\begin{itemize}
	\item the length of every $\mathcal{B}$-narrow base is smaller than a constant $K_0$,
	\item for every admissible word $W$ with base from $\mathcal{B}$ the width of every computation $W \to W_1 \to \dots \to W_t$ does not exceed
	$C(|W| + |W_t| + log_2 t / log_2 log_2 t)$ for some constant $C$.
\end{itemize}
\end{Lem}

\begin{Lem} \label{lemmu}
Let $W_0 \to W_1 \to \dots \to W_t$ be a computation. Assume that there exists $R$ such that $Rg(g(n-1)) \leq t \leq g(g(n))$ for some integer $n$.
Then the area of the corresponding diagram without hubs $\Delta$ does not exceed $Ct(|W_0|_a + |W_t|_a)$ for a constant $C$ independent of the computation.
\end{Lem}

Let us define $G(\mathcal{S} \circ \mathcal{Z})$ as the group constructed following \cite{Sap02}.
We focus on its diagrams without hubs. We shall show that the area of such diagram does not exceed $Mn^2 log'n/log' log' n+ M log'n / log' log' n \mathcal{E}(\Delta)$,
where $log'(n)=max (log_2(n),1)$.
Looking at the proof in \cite{Ols06} of Lemma.6.2  that bounds the area of diagram by $n^2 log'n + M\mathcal{E}(\Delta)$ one can remark that it depends 
on the Lemma 4.4 (\cite{Ols06}). We shall show a variant of this Lemma based on our assumptions. The reader should remember that this Lemma is true only for
diagram without hubs.

\begin{Lem} Let $\Delta$ be a diagram without hubs of height $h \geq 1$ whose base is $\mathcal{B}$-tight and $\mathcal{B}$-covered. Then the area of 
$\Delta$ does not exceed $Ch(|W|_a + |W'|_a + log' n / log' log' n +1 )$, where $W,W'$ are the labels of its top and bottom, respectively, for some $C$.
\end{Lem}
\begin{proof}
	The proof works exactly as the proof of Lemma 4.4 (\cite{Ols06}) except that the invocation of Lemma 3.31 is replaced by our Lemma \ref{powerful-lemma}.
\end{proof}

Now going step by step through the proof of Lemma 6.2 (\cite{Ols06}) and using Lemma \ref{powerful-lemma} where it is needed one can obtain the following Lemma.
\begin{Lem} The area of a reduced diagram without hubs $\Delta$ does not exceed \\
$Mn^2 log'n/log' log'n + M log'n/ log' log'n \mathcal{E}(\Delta)$ where $n=|\Delta|$.
\end{Lem}

\cite{DBLP:journals/ijac/Olshanskii07} constructs an $\mathcal{S}$-machine that satisfies some of the above properties and then shows the following, here we modify slightly the hypothesis
of the Lemma in order to fit our context ( considering diagrams without hubs instead of any diagram).
\begin{Lem} \label{lemouf} Let the perimeter $n$ of a reduced diagram without hubs $\Delta$ satisfying $n \leq g(g(r))$ for some positive integer $r$. The area
of diagram $\Delta$ does not exceed $M(n^2 + m^2 log' m) + M\mathcal{E}(\Delta)$,\\
where $m=R(g(g(r-1))log'n$ and $R$ is the constant of Lemma \ref{lemmu}.
\end{Lem}
\begin{proof} Going step by step through the proof of Lemma 5.2 ( \cite{DBLP:journals/ijac/Olshanskii07}) and using our Lemmas when it is needed proves the result.
\end{proof}

We shall now conclude our work. Remember that we want to show that for some asymptotic cone of the group $G(\mathcal{S} \circ \mathcal{Z})$ 
diagrams without hubs do not allow to conclude the non simply connectivity. For that we will show that in the group $G(\mathcal{S} \circ \mathcal{Z})$
the diagrams without hubs satisfy the following property noticed by Gromov in \cite{Gro93}:
For every $M > 1$ there exists $k$ such that for every constant $C\geq 1$, every loop $l$ in the Cayley complex of $G$, such that $\frac{1}{C}d_m\leq l \leq Cd_m$ for any
sufficiently large $m$, bounds a disc that can be subdivided into $k$ subdiscs with perimeters at most $\frac{l}{M}$.
\begin{Rem} It is sufficient to show the statement for some $M>1$. 
Indeed subdividing the subdiscs will allow one to conclude the statement for every $M>1$.

\end{Rem}

Next we will need the following Lemma from \cite{DBLP:journals/ijac/OlshanskiiS07}.
\begin{Lem} \cite{DBLP:journals/ijac/OlshanskiiS07} \label{Olslem}
	Let $\Delta$ be a triangular map whose perimeter $n$ is at least $200$. Assume that the area of $\Delta$ does not exceed $Mn^2$.
	Then there is a $k$ depending on $M$ only, such that $\Delta=\Gamma_1 \cup \dots \cup \Gamma_k$ where $\Gamma_i$ are submaps of $\Delta$, and
	$\Gamma_i \cap \Gamma_j$ is empty or a vertex, or a simple path, and perimeter $|\partial \Gamma_i|$ is at most $n/2$ for all $i=1,\dots,k$.
\end{Lem}

\cite{DBLP:journals/ijac/OlshanskiiS07} emphasizes that there exists a group such that its Dehn's function satisfies the following property:
\begin{center}
(P1): There are sequences of positive numbers $d_i \to \infty $ and $\lambda_i \to \infty $ such that $ f(n) \leq c n^{2}$ for an arbitrary integer  
$n \in \cup_{i=1}^{\infty} [\frac{d_i}{\lambda_i},\lambda_i d_i]$ and some constant $c$. 
\end{center}

Moreover \cite{DBLP:journals/ijac/OlshanskiiS07} shows that every group satisfying the condition (P1) has an asymptotic cones which is simply connected.
In the same fashion as the proof in \cite{DBLP:journals/ijac/Olshanskii07} we show the following Lemma:

\begin{Lem}\label{lemcool} 
There are sequences of positive numbers $d_i \to \infty$ , $\lambda_i \to \infty$ and a constant $c$ such that for every 
$n \in \cup_{i=1}^{\infty} [\frac{d_i}{\lambda_i},\lambda_i d_i]$ there exists a hub-free diagram $\Delta$ of perimeter $n$ from group $G(\mathcal{S} \circ \mathcal{Z})$
such that the area of $\Delta$ does not exceed $cn^2$.
\end{Lem}
\begin{proof}
It is not difficult to see that $g(g(r-1))^2 \leq g(g(r))$. Moreover the dispersion $\mathcal{E}(\Delta)$ always satisfies $\mathcal{E}(\Delta)\leq O([\partial \Delta|^2)$.
Let $n_i=g(g(i))$. Set $d_i=(n_i)^{\frac{3}{4}}$ and $\lambda_i=n_i^\varepsilon$ with $\epsilon < \frac{1}{4}$. Then applying Lemma \ref{lemouf}
concludes the proof.
\end{proof}

Take $d_i=(n_i)^{\frac{3}{4}}$, from Lemma \ref{lemcool} and Lemma \ref{Olslem} there exists $M$ and $k$ such that for every $C$ and every hub-free diagram of perimeter 
$n, \frac{1}{C}d_m \leq n \leq C d_m, m $ sufficiently large, $\Delta$ can be subdivided into $k$ subdiagrams with perimeter at most $\frac{n}{2}$.
But since $k$ is fixed and depends on $M$ it easy to find the value satisfying the Gromov condition by further decreasing the length and increasing the number of subdiagrams.

\section{Consequences}
The method developed in \cite{Ols06} is really interesting in our context. Considering non-deterministic computations and their 
associated groups in the sense of \cite{Sap02} avoids the misbehaviour of diagrams without hubs ( again such diagrams are not corresponding to valid computations of the machine).
But there is a significant drawback when one tries to apply the method as we did. Indeed Theorem \ref{supertheorem} says that the Dehn's function of
the group constructed from a machine is equivalent to $T(n)^4$ where $T(n)$ is the time complexity of the machine.
One needs to understand how a machine satisfying our assumption is obtained. Informally the machine has an exponential time computation ( Lemma \ref{explem}) 
and a width in each computation that is bounded by some function $f(n)$. Looking at the construction in \cite{Sap02} the diagrams with the hubs relation  
correspond to computations of the $\mathcal{S}$-machine. Thus the following Lemma comes for free:
\begin{Lem} 
	The Dehn's function of $G(\mathcal{S} \circ \mathcal{Z})$ is exponential.
\end{Lem}

Thus even if the topology of the asymptotic cones will not suffer from the misbehaviour of hub-free diagrams, it will not reflect the topology
of non-deterministic polynomial time computation since the considered computations now are exponential. Anyway it could be interesting to investigate further 
the application of this method and trying to construct a metric space that does not consider the exponential steps of the computation in the diagrams containing the 
hub relation. 
Moreover the formal construction of a machine satisfying the Lemmas \ref{explem} and \ref{powerful-lemma} in a group considering the hub relation is a real challenge.

\bibliography{biblio}
\bibliographystyle{eptcs}
\end{document}